\def\notes{1}

\documentclass[11pt]{article}
\usepackage[T1]{fontenc}
\usepackage[latin9]{inputenc}
\usepackage{amsthm}
\usepackage{amssymb}
\usepackage{times}
\usepackage{fullpage}
\usepackage{algorithm}
\usepackage{algorithmic}
\usepackage{hyperref}
\usepackage{amsmath}
\usepackage{cite}
\usepackage[framemethod=tikz]{mdframed}
\usepackage{dsfont}
\newtheorem{thm}{Theorem}
\newtheorem{lem}[thm]{Lemma}
\newtheorem{claim}[thm]{Claim}
\newtheorem{corollary}[thm]{Corollary}

\newtheorem{defn}[thm]{Definition}

\newenvironment{theorem-repeat}[1]{\begin{trivlist}
\item[\hspace{\labelsep}{\bf\noindent Theorem~\ref{#1} }]}%
{\end{trivlist}}

\global\long\def\R{\mathbb{R}}
\global\long\def\E{\mathbb{E}}

\global\long\def\S{P}
\global\long\def\M{M}
\global\long\def\NM{\bar{\M}}
\global\long\def\barS{V\backslash \S}
\newcommand{\MIS}{MIS~}
\global\long\def\R{R}
\global\long\def\C{C}

\usepackage[textsize=tiny]{todonotes}

\ifnum\notes=1
\newcommand{\ktodo}[1]{\textcolor{violet}{Keren: #1}}
\newcommand{\zk}[1]{\textcolor{blue}{ZK: #1}}
\newcommand{\enote}[1]{\textcolor{red}{Elad: #1}}
\else
\newcommand{\ktodo}[1]{}
\newcommand{\zk}[1]{}
\newcommand{\enote}[1]{}
\fi

\newcommand{\Gold}{G_{\mathrm{old}}}
\newcommand{\Gnew}{G_{\mathrm{new}}}
\newcommand{\ID}{\ell}

\usepackage{float}

\floatstyle{ruled}

\begin{document}
\begin{titlepage}
\title{Optimal Dynamic Distributed MIS} 

\author{Keren Censor-Hillel\thanks{Department of Computer Science, Technion. Supported in part by ISF grant 1696/14. \texttt{ckeren@cs.technion.ac.il}.}
\and
Elad Haramaty\thanks{College of Computer and Information Science, Northeastern University.  Supported in part by NSF grant CCF-1319206. Part of this work was done while the author was at Yahoo Labs. \texttt{eladh@cs.technion.ac.il}.}
\and
Zohar Karnin\thanks{Yahoo Labs. \texttt{zkarnin@ymail.com}.}
}

\maketitle
\begin{abstract}
Finding a maximal independent set (MIS) in a graph is a cornerstone task in distributed computing. The local nature of an MIS allows for fast solutions in a static distributed setting, which are \emph{logarithmic} in the number of nodes or in their degrees [Luby 1986, Ghaffari 2015]. By running a (static) distributed MIS algorithm after a topology change occurs, one can easily obtain a solution with the same complexity also for the \emph{dynamic} distributed model, in which edges or nodes may be inserted or deleted.

In this paper, we take a different approach which exploits locality to the extreme, and show how to update an MIS in a dynamic distributed setting, either \emph{synchronous} or \emph{asynchronous}, with only \emph{a single adjustment}, meaning that a single node changes its output, and in a single round, in expectation. These strong guarantees hold for the \emph{complete fully dynamic} setting: we handle all cases of \emph{insertions} and \emph{deletions}, of \emph{edges} as well as \emph{nodes}, \emph{gracefully} and \emph{abruptly}. This strongly separates the static and dynamic distributed models, as super-constant lower bounds exist for computing an MIS in the former.

We prove that for any deterministic algorithm, there is a topology change that requires $n$ adjustments, thus we also strongly separate deterministic and randomized solutions.

Our results are obtained by a novel analysis of the surprisingly simple solution of carefully simulating the greedy \emph{sequential} MIS algorithm with a random ordering of the nodes.
As such, our algorithm has a direct application as a $3$-approximation algorithm for correlation clustering. This adds to the important toolbox of distributed graph decompositions, which are widely used as crucial building blocks in distributed computing.

Finally, our algorithm enjoys a useful \emph{history-independence} property, which means that the distribution of the output structure depends only on the current graph, and does not depend on the history of topology changes that constructed that graph. This means that the output cannot be chosen, or even biased, by the adversary, in case its goal is to prevent us from optimizing some objective function.
Moreover, history independent algorithms compose nicely, which allows us to obtain history independent coloring and matching algorithms, using standard reductions. 
\end{abstract}
\thispagestyle{empty}
\end{titlepage}

\section{Introduction}
Dynamic environments are very common in distributed settings, where nodes may occasionally join and leave the network, and communication links may fail and be restored.
This makes solving tasks in a dynamic distributed setting of fundamental interest: Indeed, it is a widely studied area of research, especially in the context of asynchronous self-stabilization~\cite{Dijkstra74,Schneider93,GuellatiK10,Dolev2000}, and also in the context of severe graph changes, whether arbitrary~\cite{KuhnLO10} or evolving randomly~\cite{AvinKL08}.
In this paper, we consider a dynamic distributed setting which is \emph{synchronous} and assumes topology changes with sufficient time for recovery in between, as is typically assumed in the literature on \emph{sequential dynamic algorithms}~\cite{Demetrescu2010}.

Solutions for problems from the static distributed setting translate nicely into our dynamic distributed setting, by running them in response to topology changes, in order to adjust the output~\cite{AwerbuchS88,AwerbuchV91,LenzenSW09}. This can be quite efficient especially for \emph{local} problems, such as finding a maximal independent set (MIS) in the network graph. The cornerstone MIS problem admits fast distributed solutions whose complexities are logarithmic in the size of the graph or in the degrees of the nodes~\cite{Luby86,AlonBI86,IsraelI86,Ghaffari15}.

In this paper, we exploit locality to the extreme, and present an MIS algorithm for the dynamic distributed setting, both \emph{synchronous} or \emph{asynchronous}, which requires only \emph{a single adjustment}, where the adjustment measure of an algorithm is the number of nodes that need to change their output in response to the topology change, and a single round, in expectation. These strong guarantees hold for the \emph{complete fully dynamic} setting, i.e., we handle all cases of \emph{insertions} and \emph{deletions}, of \emph{edges} as well as \emph{nodes}, \emph{gracefully} and \emph{abruptly}.\footnote{See definitions of topology changes in Section~\ref{sec:setup}.} This is a strong separation between the static and dynamic distributed models, as super-constant lower bounds exist for the static setting~\cite{Linial92,KuhnMW04}.
We further prove that for any deterministic algorithm, there is a topology change that requires $n$ adjustments, with $n$ being the number of nodes, thus we also strongly separate deterministic and randomized solutions.
Below, we overview our technique and the applications of our result.

\subsection{Our Contribution}
Our approach is surprisingly simple: We simulate the greedy \emph{sequential} algorithm for solving MIS. The greedy sequential algorithm orders the nodes and then inspects them by increasing order. A node is added to the MIS if and only if it does not have a lower-order neighbor already in the MIS. We consider \emph{random greedy}, the variant of greedy in which the order is chosen uniformly at random. Consider simulating random greedy in a dynamic environment with the following template (ignoring the model of computation/communication for the moment). Each node needs to maintain the invariant that its state depends only on the states of its neighbors with lower order, such that it is in the MIS if and only if none of its lower order neighbors are in the MIS. When a change occurs in the graph, nodes may need to change their output, perhaps more than once, until they form a new MIS. Our key technical contribution is in proving:

\begin{theorem-repeat}{thm:ES-const}
For any \emph{arbitrary} change in the graph, the expectation over all random orders, of the number of nodes that change their output in the above random greedy template is at most 1.
\end{theorem-repeat}

\textbf{The Challenge:} We denote by $\pi$ the random order of nodes, we denote by $v^*$ the only node (if any) for which the above invariant does not hold after the topology change, and we denote by $S$ the set of nodes that need to be changed in order for the invariant to hold again at all nodes. We look at $S'$, the set of nodes that would have needed to be changed if the order was as in $\pi$, except for pushing $v^*$ to be the first node in that order. The definition of $S'$ does not depend on the real order of $v^*$ in $\pi$. Therefore, we can prove that $S$ can either be equal to $S'$ if the order of $v^*$ in $\pi$ is minimal in $S'$, and empty otherwise. Now the question is, what is the probability, given $S'$, that $v^*$ is indeed its minimal order node? The answer is that if $S'$ were deterministic, i.e.\ independent of $\pi$, the probability would be $1/|S'|$. However, $S'$ is a random set and having knowledge of its members restricts $\pi$ to be non-uniform, which in turn requires a careful analysis of the required probability. To overcome this issue, we prove that the information that $S'$ gives about $\pi$ is either about the order between nodes not in $S'$, or about the order within $S'\backslash\{v^*\}$, which both do not effect the probability that $v^*$ is the minimal in $S'$.

\textbf{Distributed Implementation:}
This powerful statement of $\E[|S|] \leq 1$ directly implies that a \emph{single} adjustment is sufficient for maintaining an MIS in a dynamic setting. A direct distributed implementation of our template implies that in expectation also a single round is sufficient. This applies both to the synchronous and asynchronous models, where the number of rounds in the asynchronous model is defined as the longest path of communication.

\textbf{Obtaining $O(1)$ Broadcasts and Bits:}
In fact, in the synchronous model, it is possible to obtain an expected number of $O(1)$ \emph{broadcasts} and \emph{bits}. Here the number of broadcasts is the total number of times, over all nodes, that any node sends a $O(\log{n})$-bit broadcast message (to all of its neighbors)\footnote{We emphasize that the term broadcast is used here to indicate the more restricted setting of not being able to send different messages to different neighbors in the same round. It does not refer to a wireless setting of communication.}. Moreover, since we only need a node to know the order between itself and its neighbors, using a similar technique to that of~\cite{MetivierRSZ11}, we can obtain that in expectation, a node only needs to send a constant number of bits in each broadcast. The above holds for edge insertions and deletions, graceful node deletion, and node \emph{unmuting}, while for an abrupt deletion of a node $v^*$ we will need $O\left(\min\{\log(n), d(v^*)\}\right)$ broadcasts, and for an insertion of a node $v^*$ we will need $O(d(v^*))$ broadcasts, in expectation. 

This is done with a careful dynamic distributed implementation which guarantees that each node that changes its output does so at most $O(1)$ times, as opposed to the direct distributed implementation \footnote{This bears some similarity to the method in~\cite{Turau07}, where the number of \emph{moves} is reduced in an MIS self-stabilizing algorithm by adding a possible \emph{wait} state to the standard \emph{in MIS} and \emph{not in MIS} states.}. Hence, obtaining these broadcast and bit complexities comes at a cost of increasing the round complexity, but it remains constant (albeit not 1).
In what follow, we focus on this result since we find it intriguing that we can get as little as $O(1)$ \emph{total communication}, while paying $O(1)$ rounds instead of a single round is arguably not a big cost. 
\textbf{Matching Lower Bounds:}
We claim that any deterministic algorithm requires $n$ adjustments, which can be seen through the following example. Let $A$ be a dynamic deterministic MIS algorithm. Let $G_0$ be the complete bipartite graph over two sets of nodes of size $k$. We denote by $L$ the side of $G_0$ that is chosen to be the MIS by $A$, and we denote the other side by $R$. For every $i\in[k]$ let $G_i$ be the graph obtained after deleting $i$ nodes from $L$, and consider executing $A$ on $G_0,G_1,...,G_k$. For every $i$, since $G_i$ is a complete bipartite graph, one of the sides has to be the MIS. Since $G_k$ contains only disconnected nodes of $R$ then $R$ is the only MIS of $G_k$. This implies that after some specific change along the sequence, the side of the MIS changes from $L$ to $R$. In this topology change, \emph{all} of the nodes change their output.

This gives a strong separation between our result and deterministic algorithms.
Moreover, it shows that (1) the expected adjustment complexity of any algorithm must be at least 1, as we have a sequence of $k$ topology changes that lead to at least $k$ adjustments,  and (2) it is impossible to achieve high probability bounds that improve upon a simple Markov bound. Specifically, this explains why we obtain our result in expectation, rather than with high probability. This is because the example can be inserted into any larger graph on $n$ nodes, showing that \emph{for every} value of $k$, there exists an instance for which at least $\Omega(k)$ adjustments are needed with probability at least $1/k$.

\textbf{Approximate Correlation Clustering:} In addition to the optimal complexity guarantees, the fact that our algorithm simulates the random greedy sequential algorithm has a significant application to correlation clustering. Correlation clustering requires the nodes to be partitioned into clusters in a way that minimizes the sum of the number of edges outside clusters and the number of non neighboring pairs of nodes within clusters (that is, missing edges within clusters). Ailon et al.~\cite{AilonCN08} show that random greedy obtains a $3$-approximation for correlation clustering\footnote{In the same paper, they also provide a $2.5$ approximation based on rounding a solution of a linear program. We do not elaborate on the details of this algorithm, nor the history of the correlation clustering problem as it is outside the scope of our paper.}, by having each MIS node inducing a cluster, and each node not in the MIS belonging to the cluster induced by the smallest random ID among its MIS neighbors. This directly translates to our model, by having the nodes know that random ID of their neighbors. Graph decompositions play a vital role in distributed computing (see, e.g.,~\cite{Peleg2000}), and hence the importance of obtaining a $3$-approximation for correlation clustering.


\textbf{History Independence:} Finally, our algorithm has a useful property, which we call \emph{history independence}, which means that the structure output by the algorithm (e.g., the MIS) depends only on the current graph, and does not depend on the history of topology changes. This means that the output cannot be chosen, or even biased, by the adversary, in case its goal is to prevent us from optimizing some objective function.
Moreover, history independent algorithms compose nicely, which allows us to obtain history independent coloring and matching algorithms, using standard reductions.

\subsection{Related Work}
\label{subsec:related}

\textbf{Distributed MIS:} Finding an MIS is a central theme in the classical distributed setting. The classic algorithms~\cite{Luby86,AlonBI86,IsraelI86} complete within $O(\log{n})$ rounds, with high probability. More recently, a beautiful line of work reduced the round complexity to depend on $\Delta$, the maximal degree in the graph. These include the $O(\Delta+\log^*{n})$-round algorithm of~\cite{BarenboimEK14}, the $O(\log{\Delta}\sqrt{\log{n}})$-round algorithm of~\cite{BarenboimEPS12}, and the very recent $O(\log{\Delta})+2^{O(\sqrt{\log\log{n}})}$-round algorithm of~\cite{Ghaffari15}. An excellent source for additional background and relations to coloring can be found in~\cite{2013Barenboim}.

\textbf{Distributed dynamic MIS:} The problem of finding a fast dynamic distributed MIS algorithm appears as an open problem in~\cite{Elkin07}, which studies the problem of maintaining a sparse spanner in this setting. Additional problems in this setting are also addressed in~\cite{BaswanaKS12}, and in slightly different settings in~\cite{CiceroneSFN03,AwerbuchCK90,Italiano91,KormanP03}. However, we are unaware of any other work in this setting about maintaining an MIS.
One standard approach for maintaining an MIS is running distributed algorithms that are designed for the static setting. This can be done for any distributed algorithm, sometimes using a corresponding \emph{compiler}, e.g., when applied to an asynchronous dynamic setting~\cite{AwerbuchS88,AwerbuchV91,LenzenSW09}. One important exception is the solution in~\cite{KonigW13}, which as in our algorithm, requires a constant number of rounds, but as opposed to our algorithm, makes the strong assumptions that (1) a node gracefully departs the network, and (2) messages may have unbounded size. An additional difference is that the number of broadcasts, as opposed to the number of rounds, may be large.

\textbf{Additional distributed dynamic algorithms:} A huge amount of literature is devoted to devising different algorithms in a self-stabilizing setting (see, e.g.,~\cite{Schneider93,GuellatiK10,Dolev2000} and references therein). This setting is inherently different from ours since it measures the time it takes an algorithm to reach a correct output starting from \emph{any} arbitrary configuration. The setting is asynchronous, but considers a notion of time that is different than ours, where an asynchronous round requires that each node communicates with all of its neighbors. This inherently implies a lot of communication (broadcasts).


An MIS-based clustering algorithm for the asynchronous model that appeared in~\cite{DolevT09} also uses a random node order for recovering after a change. However, their self-stabilizing setting differs from ours in several aspects, such as assuming a bounded degree graph and discussing corrupted states of multiple nodes, and multiple topology changes. In addition, our techniques and analysis are completely different. In particular, the clustering obtained there may not be an approximation to correlation clustering.
Furthermore, the number of rounds required by ~\cite{DolevT09} is $O(\log(n))$ as opposed to the single round algorithm (in expectation) presented here.

Related, but not identical, notions of error confinement, fault local and fault amendable algorithms have been studied in~\cite{KuttenP99,KuttenP00,AzarKP2010}, where the internal memory of a node may change.
Another property that self-stabilizing algorithms should aim for is super-stabilization~\cite{DolevH97}, which means that they are self-stabilizing (eventually output the required structure) and also recover extremely fast from \emph{a single} topology change. Super-stabilization requires also a small adjustment measure, which is the maximum number of nodes that have to change their output. Our MIS algorithm recovers from a single topology change in a single round, and has an adjustment measure of exactly 1, in expectation.

\textbf{Simulating the sequential greedy algorithm:} Simulating random greedy has been used before in order to obtain fast solutions for sequential local computation algorithms (LCA). In this setting, the algorithm does not have access to the entire graph, but rather an oracle access using queries about nodes or edges, and needs to provide an approximate solution for various problems, among which are the problems considered in this paper. We emphasize that the models are inherently different, and hence is our technical analysis. While we bound the size of the set of nodes that may change their output after a topology change, studies in the local computation literature~\cite{NguyenO08,YoshidaYI12,MansourRVX12,LeviRY15} bound the size of the set of nodes that need to be recursively queried in order to answer a \emph{random} node query. In some sense, these sets are opposite: We begin with a single node $v$ changing its state due to a topology change, and look at the set of nodes that change their state \emph{due} to the change of $v$. Local computation algorithms begin with a single node $v$ and look at the set of nodes whose states \emph{determine} the state of $v$.


\section{Dynamic Distributed Computations}
\label{sec:setup}
The distributed setup is a standard message passing model. The network is modeled by an undirected graph $G=(V,E)$ where the node set is $V$, and $E$ consists of the node pairs that have the ability to directly communicate.
We assume a broadcast setting where a message sent by a node is heard by all of its neighbors. Also, we assume a synchronous communication model, where time is divided into rounds and in each round any willing node can broadcast a message to its neighbors. We restrict the size of each message to be $O(\log(n))$ bits, with $n=|V|$ being the size of the network\footnote{This is the standard assumption in a distributed setting. In our dynamic setting where the size of the graph may change we assume knowledge of some upper bound $N \geq n$, with $N = n^{O(1)}$, and restrict the message length to $O(\log(N))= O(\log(n))$.}.
The computational task is to maintain a graph structure, such as a maximal independent set (MIS) or a node clustering. That is, each node has an output, such that the set of outputs defines the required structure.

Our focus is on a dynamic network, where the graph changes over time. As a result, nodes may need to communicate in order to adjust their outputs. The system is \emph{stable} is when the structure defined by the outputs satisfies the problem requirements.

A graph topology change can be with respect to either an edge or a node. In both cases we address both deletions and insertions, both of which are further split into two different types.
For deletions we discuss both a \emph{graceful deletion} and an \emph{abrupt deletion}. In the former, the deleted node (edge) may be used for passing messages between its neighbors (endpoints), and retires completely only once the system is stable again. In the latter, the neighbors of the deleted node simply discover that the node (edge) has retired but it cannot be used for communication. For insertions, we distinguish between a \emph{new node insertion} and an \emph{unmuting} of a previously existing node. In the former, a new node is inserted to the graph, possibly with multiple edges. In the latter, a node that was previously invisible to its neighbors but heard their communication,
becomes visible and enters the graph topology\footnote{The distinction is only relevant for nodes insertions, as there is no knowledge associated with an edge.}.


We assume that the changes are infrequent so that they occur in large enough time gaps, so that the system is always stable before a change occurs.
We consider the performance of an algorithm according to three complexity measures. The first is the \emph{adjustment-complexity}, measuring the number of nodes that change their output 
as a result of the recent topology change. The second is the \emph{round-complexity}, which is the number of rounds required for the system to become stable. Finally, the third, more harsh, score is the \emph{broadcast-complexity}, measuring the total number of broadcasts.

Our algorithms are randomized and thus our results apply to the expected values of the above measures, where the expectation is taken over the randomness of the nodes. We emphasize that this is the only randomness discussed; specifically, the result is not for a random node in the graph nor a random sequence of changes, but rather applies to any node and any sequence of changes. It holds for \emph{every} change in the graph, not only in amortized over all changes.
For this, we make the standard assumption of an \emph{oblivious} non-adaptive adversary. This means that the topology changes do not depend on the randomness of the algorithm. This standard assumption in dynamic settings is natural also for our setting, as, for example, an adaptive adversary can always choose to delete MIS nodes and thereby force worst-case behavior in terms of the number of adjustments.

In what follows we discuss the problem of computing an \emph{MIS}. Here, the outputs of the nodes define a set $M$, where any two nodes in $M$ are not connected by an edge, and any node not in $M$ has a neighbor in $M$. The second problem we discuss is that of \emph{correlation clustering}. Here, the objective is to find a partitioning $\cal{C}$ of the node set $V$, where we favor partitions with a small number  of ``contradicting edges". That is, we aim to minimize the sum
$ \sum_{C \in {\cal C}} \sum_{u,v \in C}  \mathds{1}_{[(u,v) \notin E]} + \sum_{C_1 \neq C_2 \in {\cal C}} \sum_{u\in C_1, v \in C_2}  \mathds{1}_{[(u,v) \in E]} $.

\section{A Template for Maintaining a Maximal Independent Set}
In this section we describe a template for maintaining a maximal independent set (MIS). Initially, we are given a graph $G=(V,E)$ along with an MIS that satisfies certain properties, and after a topology change occurs in the graph, applying the template results in an MIS that satisfies the same properties. That is, the template describes what we do after a single topology change, and if one considers a long-lived process of topology changes, then this would correspond to having initially an empty graph and maintaining an MIS as it evolves. We emphasize that the template describes a process that is not in any particular model of computation, and later in Section~\ref{sec:dist_impl} we show how to implement it efficiently in our dynamic distributed setting. This also means that there are only four topology changes we need to consider: edge-insertion, edge-deletion, node-insertion and node-deletion. For example, the notions of abrupt and graceful node deletions are defined with respect to the dynamic distributed setting because they affect communication, and therefore the implementation of the template will have to address this distinction, but the template itself is only concerned with a single type of node deletion, not in any particular computation model.

Throughout, we assume a uniformly random permutation $\pi$ on the nodes $v \in V$. We define two \emph{states} in which each node can be: $\M$  for an MIS node, and $\NM$ for a non-MIS node. We abuse notations and also denote by $\M$ and $\NM$ the sets of all MIS and non-MIS nodes, respectively. Our goal is to maintain the following \emph{\MIS invariant}: A node $v$ is in $\M$ if and only if all of its neighbors $u \in N(v)$ which are ordered before it according to $\pi$, i.e., for which $\pi(u) < \pi(v)$, are not in $\M$. It is easy to verify that whenever the \MIS invariant is satisfied, it holds that the set $\M$ is a maximal independent set in $G$.
Furthermore, it is easy to verify that this invariant simulates the greedy sequential algorithm, as defined in the introduction.

When any of the four topology changes occurs, there is at most a single node for which the MIS invariant no longer holds. We denote this node by $v^* = v^*(\Gold,\Gnew,\pi)$, where $\Gold$ and $\Gnew$ are the graphs before and after the topology change. For an edge insertion or deletion, $v^*$ is the endpoint with the larger order according to $\pi$. For a node insertion or deletion, $v^*$ is the  node. \footnote{For a node deletion, we slightly abuse the definition of $v^*$ in order to facilitate the presentation, and consider it to be the deleted node. This means that here we consider an intermediate stage of having $v^*$ still belong to the graph w.r.t. the \MIS invariant of all the other nodes, but for $v^*$ the \MIS invariant no longer holds. This is in order to unify the four cases, otherwise we would have to consider all of the neighbors of a deleted node as nodes for which the \MIS invariant no longer holds after the topology change.}
In case the topology change is an edge change, we will need also to take into consideration its other endpoint. We denote it by $v^{**} =v^{**}(\Gold,\Gnew,\pi)$, and notice that by our notation, it must be the case that $\pi(v^{**}) < \pi(v^{*})$. In order to unify our proofs for all of the four possible topology changes, we talk about a node $v^{**}$ also for node changes. In this case we define $v^{**}$ to be $v^*$ itself, and we have that $\pi(v^{**}) = \pi(v^{*})$. Therefore, for any topology change, it holds that $\pi(v^{**}) \leq \pi(v^{*})$.

To describe our template, consider the case where a new edge is inserted and it connects two nodes $\pi(v^{**}) < \pi(v^{*})$, where both nodes are in $\M$. As a result, $v^*$ must now be deleted from the MIS and hence we need to change its state. Notice that as a result of the change in the state of $v^*$, additional nodes may need their state to be changed, causing multiple state changes in the graph.
An important observation is that it is possible that during this process of propagating local corrections of the \MIS invariant, we change the state of a node more than once. As a simple example, consider the case in which $v^*$ has two neighbors, $u_1$ and $u_2$, for which $\pi(v^*) < \pi(u_1),\pi(u_2)$, and that $u_1$ and $u_2$ are connected by a path $(u_1,w_1,w_2,u_2)$, with $\pi(u_1) < \pi(w_1) < \pi(w_2) < \pi(u_2)$. Now, when we change the state of $v^*$ to $\NM$, both $u_1$ and $u_2$ need to be changed to $\M$, for the \MIS invariant to hold. This implies that $w_1$ needs to be changed to $\NM$ and $w_2$ needs to be changed to $\M$. In this case, since $\pi(w_2) < \pi(u_2)$, the node $u_2$ needs to be changed back to state $\NM$.

The above observation leads us to define a set of \emph{influenced} nodes, denoted by $S=S(\Gold,\Gnew,\pi)$, containing $v^*$ in the scenario where we need to change its state, and all other nodes whose state we must  subsequently change as a result of the state change of $v^*$.
To formally define the set $S$ we introduce some notations. The notations rely on the graph structure of $\Gnew$ unless the change is a node deletion in which case the rely on $\Gold$.
For each node $u$, we define $I_{\pi}(u) = \{ v \in N(u) ~|~ \pi(v) <  \pi(u)\}$, the set of neighbors of $u$ that are ordered before it according to $\pi$.
These are the nodes that can potentially \emph{influence} the state of $u$ according to the \MIS invariant. 
The definition of $S$ is recursive, according to the ordering induced by $\pi$.
If immediately after the topology change, in the new graph $G$ with the order $\pi$ it holds that the \MIS invariant still holds for $v^*$, then we define $S=\emptyset$. (This is motivated by the fact that no node is influenced by this change.) Otherwise, we denote $S_0 = \{v^*\}$, and inductively define
\begin{equation} \label{eq:defS}
S_{i} = \{u ~|~ u \in \M \mbox{, and } S_{i-1}\cap I_{\pi}(u)\neq\emptyset\} \cup \{u ~|~ u \in \NM \mbox{, and every }  v \in I_{\pi}(u) \cap \M \mbox{ is in }\cup_{j=0}^{i-1}{S_j}) \}.
\end{equation}

The set $S$ is then defined as $S=\bigcup_i{S_i}$. Notice that a node $u$ can be in more than one set $S_i$, as is the case for $u_2$ in the example above, which is in both $S_1$ and $S_4$. The impact of a node $u$ being in more than one $S_i$ is that in order to maintain the \MIS invariant, we need to make sure that we update the state of $u$ after we update that of $w$, for any $w$ such that $w \in I_{\pi}(u)$. Instead of updating the state of $u$ twice, we can simply wait and update it only after the state of every such $w$ is updated. For this, we denote by $i_u = \max\{i~|~ u \in S_i\}$ the maximal index $i$ for which $u$ is in $S_i$.

\begin{algorithm}
Initially, $G=(V,E)$ satisfies the \MIS invariant.\\
On topology change at node $v^*$ do:\\
1. Update state of $v^*$ if required for \MIS to hold\\
2. For $i \leftarrow 1$, until $S_{i} = \emptyset$, do:\\
3. \quad For every $u \in S_{i}$ such that $i = i_u$:\\
4. \quad\quad Update state of $u$\\
5. \quad $i \leftarrow i+1$
\caption{A Template for Dynamic Correlation Clustering.}
\label{alg:template}
\end{algorithm}

We formally describe our template in Algorithm~\ref{alg:template}.
By construction, the updated states after executing Algorithm~\ref{alg:template} satisfy the \MIS invariant.
In addition, the crucial property that is satisfied by the above template is that in expectation, the size of the set $S$ is $1$. The remainder of this section is devoted to proving the following, which is our main technical result.

\begin{thm}
\label{thm:ES-const}
For every two graphs $\Gold$ and $\Gnew$ that differ only by a single edge or a single node, it holds that  
$\E_\pi \left[|S(\Gold, \Gnew, \pi)|\right] \leq 1$. 
\end{thm}

\paragraph{Outline of the proof:} In order to prove that $\E[|S|] \leq 1$, 
instead of analyzing the set $S$ directly, we analyze the set $S' = S'(\Gold,\Gnew,\pi,v^*)$,  which is defined via recursion similarly to $S$ with three modifications: (1) It is always the case that $S_0' = \{v^*\}$ (2) The graph according to which $S'$ is defined is $\Gold$ in the case of a node deletion or an edge insertion, and $\Gnew$ otherwise. (3) The permutation according to $S'$ is defined as $\pi'$, that is identical to $\pi$ other than its value for $v^*$ that is forced to be the minimal among all other $\pi$ values. Notice that $S'$ does not depend on $\pi(v^*)$ and in particular, having knowledge about its elements does not give any information as to whether $\pi(v^*) < \pi(v^{**})$ or vice versa. 


In Lemma~\ref{lem:SandS'}, we prove that if $\pi(v^*) \neq \min\left\{ \pi(u)\mid u\in S'\right\}$ then $S=\emptyset$, and otherwise $S=S'$ (in fact, it would be enough that $S\subseteq S'$). Then, in Lemma~\ref{lem:one-over-P}, we prove that for any set $P \subseteq V$, given the event that $P=S'$, the probability, over the random choice of $\pi$, that $\pi(v^*) =\min\left\{ \pi(u)\mid u\in P\right\}$ is $1/|P|$. This leads to the required result of Theorem~\ref{thm:ES-const}.
Lemma~\ref{lem:one-over-P} would be trivial if there was no correlation between $\pi$ and $S'$. However, the trap we must avoid here is that $S'$ is defined according to $\pi$, and therefore when analyzing its size we cannot treat $\pi$ as a \emph{uniformly} random permutation. To see why, suppose we know that inside $S'\setminus\{v^*\}$ we have nodes with large order in $\pi$. Then the probability that the order of $v^*$ in $\pi$ is smaller than all nodes in $S'\setminus\{v^*\}$, is much larger than $1/|S'|$, and can in fact be as large as $1-o(1)$. In other words, $S'$ gives some information over $\pi$. Nevertheless, we show that this information is either about the order between nodes outside of $S'$, or about the order between nodes within $S'\setminus\{v^*\}$. Both types of restrictions on $\pi$ do not affect the probability that $v^*$ is the minimal of $S'$.


We now formally prove our result as outlined above. Throughout we use the notation $u \in \M$ or $u \in \NM$. This applies only to nodes $u$ for which we are guaranteed that their states remain the same despite the topology change.
\begin{lem}
\label{lem:SandS'}
If $\pi(v^*)\neq\min\left\{ \pi(u)\mid u\in S'\right\} $ then $S=\emptyset$. Otherwise, $S\subseteq S'$.
\end{lem}
\begin{proof}
First, assume that $\pi(v^*)\neq\min\left\{ \pi(u)\mid u\in S'\right\}$. We show that the \MIS invariant still holds after the topology change, and so $S=\emptyset$.
Consider the node $w$, for which $\pi(w)=\min\left\{ \pi(u)\mid u\in S'\right\}$. Notice that $w\not\in S$, because $\pi(w) < \pi(v^*)$. We claim that $w \in \M$. Assume, towards a contradiction, that $w \in \NM$. This implies that $w$ has a neighbor $u \in \M$ such that $\pi(u) < \pi(w)$. For this node $u$ we must have $u \notin S'$ due to the minimality of $\pi(w)$. It follows, according to the construction of $S'$ that $w$ cannot be an element of $S'$, leading to a contradiction.

We have that $w \in \M$ and due to the minimality of $\pi(w)$, it must be that $w \in S'_1$, which implies that $w$ is a neighbor of $v^*$. But then, when considering $S$, $v^*$ has a neighbor other than $v^{**}$ which is ordered before it according to $\pi$ which is in $\M$. In the case of an edge insertion or deletion, this means that $v^{*}$ remains in $\NM$ despite the topology change meaning that $S=\emptyset$. In the case of a node deletion, $v^*$ was not in $\M$ prior to the change hence $S=\emptyset$. In the case of a node insertion, $v^*$ does not enter $\M$ hence again, $S=\emptyset$.

Next, assume that $\pi(v^*)=\min\left\{ \pi(u)\mid u\in S'\right\}$. We show that either $S=\emptyset$ or $S=S'$.
%
If there is no need to change the state of $v^*$ as a result of the topological change then $S_0=\emptyset$, and so $S=\emptyset$ and the claim holds. It remains to analyze the case where $S_0=S'_0=\{v^*\}$. 
If $u \in S'_1$ then $\pi(v^*) < \pi(u)$ hence according to its definition $u \in S_1$. If $u \notin S'_1$ then $u$ must have a neighbor $w \in \M$ with $\pi(w)<\pi(u)$ meaning that $u \notin S_1$. We have that $S_1=S'_1$ and similarly $S_i=S'_i$ for all $i >1$. We conclude that $S'=S$ as required.
\end{proof}

The following lemma shows that the probability of having $S=S'$ is $1/|S'|$, which immediately lead to Theorem~\ref{thm:ES-const} as the only other alternative is $S=\emptyset$. 
\begin{lem}
\label{lem:one-over-P}
For any set of nodes $\S\subseteq V$, it holds that $$\Pr\left[\pi(v^*)=\min\left\{ \pi(u)\mid u\in \S\right\} \mid S'=\S \mbox{ and } \pi(v^{**})\leq\pi(v^{*}) \right]=\frac{1}{|\S|}.$$
\end{lem}

To prove this lemma we focus on $S'$. Notice that the events we considered in the previous lemma depend only on the ordering 
implied by $\pi$ and hold for any configuration of states for the nodes that satisfy the \MIS invariant. Roughly speaking, the lemma 
will follow from the fact the the event $S'=\S$ does not give any information about the order implied by $\pi$ between nodes in $\S$ 
and nodes in $\barS$. To this end, for every permutation $\tau$ on $V$, we define $S'(\tau) = S'(\Gold,\Gnew,\tau,v^*)$ as the set corresponding to $S'$ under the ordering induced by $\tau$. 
We denote by $\Pi_{\S}$ the set of all permutations $\tau$ for which it holds that $S'(\tau)=\S$. We first need to establish the 
following about permutations in $\Pi_{\S}$: If $\pi$ and $\sigma$ are two permutations on $V$ such that $\pi|_{\S}=\sigma|_{\S}$ 
and $\pi|_{\barS}=\sigma|_{\barS}$, then $\sigma\in\Pi_{\S}$ if and only if $\pi\in\Pi_{\S}$.

\begin{claim}
\label{claim:u-in-barS}
Let $P \subseteq V$ be a set of nodes, and let $\pi$ and $\sigma$ be two permutations such that $\pi|_{\S}=\sigma|_{\S}$ and $\pi|_{\barS}=\sigma|_{\barS}$. Assume $\pi \in \Pi_\S$. We have that $\barS \subseteq V \setminus S'(\sigma)$ and every $u \in \barS$ has the same state according to $\pi$ and $\sigma$.
\end{claim}

\begin{proof}
Let $u \in \barS$. We prove that $u \in V \setminus S'(\sigma)$ and that its state under $\sigma$ is the same as it is under $\pi$ by induction on the order of nodes in $\barS$ according to $\pi$ (which is equal to their order according to $\sigma$).

For the base case, assume that $u$ has the minimal order in $\barS$.
We claim that $u$ cannot have a neighbor in $\S$.
Assume, towards a contradiction, that $u$ has a neighbor $w\in\S$. Since $w\in\S$ then it is possible that after the a $w$ will be in $\M$. Since two nodes in $\M$ cannot be neighbors and $u \not\in\S$, then $u$ must be in $\NM$ according to $\pi$. In this case there is a node $z\in I_{\pi}(u)\cap\barS$ that is in $\M$ according to $\pi$. But this cannot occur due to the minimality of $\pi(u)$ in $\barS$. Therefore, $u$ has no neighbors in $\S$ as required.

We have that all of the neighbors of $u$ are in $\barS$ and that $u$ is the minimal among its neighbors according to $\pi$. Since $\pi|_{\barS}=\sigma|_{\barS}$ we have that $u$ has the minimal order among its neighbors according to $\sigma$. This translates into $u$ having a state of $M$ under $\sigma$ and in particular, $u$ is not an element of $S'(\sigma)$, thus proving our base case.
%

For the induction step, consider a node $u\in\barS$, and assume the claim holds for every $w \in \barS \cap I_{\pi}(u)$. We consider two cases, depending on whether $u$ has a neighbor in $\S$ or not.

Case 1:  $u$ does not have any neighbor in $\S$.
If $u \in \NM$, then there is a node $z\in I_{\pi}(u)\cap \barS$ that is in $\M$ according to $\pi$. By the induction hypothesis, $z \in V \setminus S'(\sigma)$ and $z \in \M$ also according to $\sigma$. Since $\pi|_{\barS}=\sigma|_{\barS}$, we have that $u$ is in $\NM$ according to $\sigma$ too.
Otherwise, if $u \in \M$, then every $w\in I_{\pi}(u)$ (which is also $\barS$) is in $\NM$ according to $\pi$.
Any node $w \in I_{\sigma}(u)$ is also in $w \in I_{\pi}(u)$, since it is not in $\S$ and $\pi|_{\barS}=\sigma|_{\barS}$. The induction hypothesis on $w$ gives that it is also in $V \setminus S'(\sigma)$ (otherwise it would be in $S'_{\pi}=\S$ in contradiction to the assumption of case 1), and its state according to $\sigma$ is $\NM$. Hence, $u$ must be in $V \setminus S'(\sigma)$ as well, and in state $\M$ according to $\sigma$.

Case 2: Assume that $u$ has a neighbor $w\in\S$. Since $w\in\S$ then it is possible that after the algorithm $w$ will be in $\M$. Since two nodes in $\M$ cannot be neighbors and $u \not\in\S$, then $u$ must be in $\NM$ according to $\pi$. In this case there is a node $z\in I_{\pi}(u)\backslash\S$ that is in $\M$ according to $\pi$. By the induction hypothesis, $z \in V \setminus S'(\sigma)$ and $z \in \M$ also according to $\sigma$. Since $\pi|_{\barS}=\sigma|_{\barS}$, we have that $u$ is in $\NM$ and in $V \setminus S'(\sigma)$ according to $\sigma$ too.
\end{proof}

\begin{claim}
\label{claim:sigma-pi}
Let $P \subseteq V$ be a set of nodes, and let $\pi$ and $\sigma$ be two permutations such that $\pi|_{\S}=\sigma|_{\S}$ and $\pi|_{\barS}=\sigma|_{\barS}$. Assume $\pi \in \Pi_\S$. We have that $\S \subseteq S'(\sigma)$.
\end{claim}
\begin{proof}
We prove that every node $u \in \S$ is also in $S'(\sigma)$ by induction on the order of nodes in $\S$ according to $\pi$ (which is equal to their order according to $\sigma$), with the modification forcing $v^*$ to be the first among the nodes of $P$.
The base case is for $v^*$, which is clearly in both sets $S'(\pi)$ and $S'(\sigma)$.
Consider a node $u \in \S$ and assume that the claim holds for every node in $\S$ which is ordered before $u$ according to $\pi$.
Since $u \in \S$ and $u \neq v^*$ there must be some $w \in I_\pi(u) \cap \S$ since $\pi|_{\S}=\sigma|_{\S}$ and $u \in \S$ we have according to our induction hypothesis that $w \in S'(\sigma)$, meaning that $I_\sigma(u) \cap S'(\sigma)$ is non-empty.

Consider now an arbitrary $w \in I_\sigma(u)$. If $w \in \S$ then since $\pi|_{\S}=\sigma|_{\S}$ and $u \in \S$ we have according to our induction hypothesis that $w \in S'(\sigma)$. If $w \notin \S$ then it must be the case that $w \in \NM$ according to $\pi$, otherwise $u$ cannot be in $\S$. We thus have according to Claim~\ref{claim:u-in-barS} that (1) $w \in V \setminus S'(\sigma)$ and (2) $w \in \NM$ according to $\sigma$. It follows that all neighbors of $u$ in $I_\sigma(u)$ are either in $S'(\sigma)$ or in $\NM$ according to $\sigma$, hence since $I_\sigma(u) \cap S'(\sigma) \neq \emptyset$ it must be the case that $u \in S'(\sigma)$.
\end{proof}

Claims~\ref{claim:u-in-barS} and~\ref{claim:sigma-pi} combined imply that if $\pi|_{\S}=\sigma|_{\S}$ and $\pi|_{\barS}=\sigma|_{\barS}$ then $\sigma \in \Pi_{\S}$ if and only if $\pi \in \Pi_{\S}$. We are now ready for the proof of Lemma~\ref{lem:one-over-P}.
\begin{proof}
\emph{(of Lemma~\ref{lem:one-over-P})~}
Given two permutations $\sigma^{+}$ and $\sigma^{-}$ on $\S\backslash\{v^*\}$ and $\barS$, respectively, we define $\rho_{\sigma^{+},\sigma^{-}}$ as 
$\rho_{\sigma^{+},\sigma^{-}} = \Pr\left[\forall u\in\S, \pi(v^*) \leq \pi(u) \mid\pi|_{\S\backslash\{v^*\}}=\sigma^{+} \mbox{ and } \pi|_{\barS}=\sigma^{-}\right].$


First, we observe that for two pairs of permutations $\sigma_{1}^{+},\sigma_{1}^{-}$ and $\sigma_{2}^{+},\sigma_{2}^{-}$ as above, it holds that $\rho_{\sigma_{1}^{+},\sigma_{1}^{-}}=\rho_{\sigma_{2}^{+},\sigma_{2}^{-}}$. This is because given the condition for $\sigma_{1}^{+},\sigma_{1}^{-}$, applying the permutation $(\sigma_{1}^{+})^{-1}\sigma_{2}^{+}$ to nodes in $\S\backslash\{v^*\}$ and applying the permutation $(\sigma_{1}^{-})^{-1}\sigma_{2}^{-}$ to nodes in $\barS$ has no affect on whether the event $\forall u\in\S, \pi(v^*)\leq\pi(u)$ holds.
Next, since $\Pr\left[\forall u\in\S, \pi(v^*)\leq\pi(u)\right]=\frac{1}{|\S|}$, we have that for any pair of permutations $\sigma^{+},\sigma^{-}$ on $\S\backslash\{v^*\}$ and $\barS$, respectively:
\vspace{-0.1in}
\begin{eqnarray*}
\frac{1}{|\S|} & = & \Pr\left[\forall u\in\S, \pi(v^*)\leq\pi(u)\right]=\sum_{\tau^{+},\tau^{-}}\rho_{\tau^{+},\tau^{-}}\Pr\left[\pi|_{\S\backslash\{v^*\}}=\tau^{+}\mbox{ and }\pi|_{\barS}=\tau^{-}\right]\\
 & = & \sum_{\tau^{+},\tau^{-}}\rho_{\sigma^{+},\sigma^{-}}\Pr\left[\pi|_{\S\backslash\{v^*\}}=\tau^{+}\mbox{ and }\pi|_{\barS}=\tau^{-}\right]=\rho_{\sigma^{+},\sigma^{-}}.
\end{eqnarray*}

Finally, Claims~\ref{claim:u-in-barS} and~\ref{claim:sigma-pi} imply that for every set $\S\subseteq V$ there is a set of $t=t_{\S}$ pairs of permutations $\{(\sigma_{1}^{+}, \sigma_{1}^{-}),\dots,(\sigma_{t}^{+}, \sigma_{t}^{-}) \}$ on $\S\backslash\{v^*\}$ and $\barS$, respectively, such that $\Pi_{\S}=\{\pi\mid \exists i, \pi|_{\S\backslash\{v^*\}}=\sigma_{i}^{+}\mbox{ and }\pi|_{\barS}=\sigma_{i}^{-}\}$. We conclude that for a given set $\S\subseteq V$:
\vspace{-0.2in}
\begin{eqnarray*}
\Pr_{\pi\in\Pi_{\S}}\left[\forall u\in\S, \pi(v^*)\leq\pi(u)\right] & = & \sum_{i=1}^{t}\rho_{\sigma_{i}^{+},\sigma_{i}^{-}}\Pr\left[\pi|_{\S\backslash\{v^*\}}=\sigma_{i}^{+}\mbox{ and }\pi|_{\barS}=\sigma_{i}^{-}\mid\pi\in\Pi_{\S}\right]
\end{eqnarray*}
\vspace{-0.25in}
\begin{eqnarray} \label{eq:vs_min}
 & = & \frac{1}{|\S|}\sum_{i=1}^{t}\Pr\left[\pi|_{\S\backslash\{v^*\}}=\sigma_{i}^{+}\mbox{ and }\pi|_{\barS}=\sigma_{i}^{-}\mid\pi\in\Pi_{\S}\right]=\frac{1}{|\S|}.
\end{eqnarray}
To complete the proof, we argue that knowing that $\pi(v^{**}) \leq \pi(v^*)$ can only decrease the probability that $\pi(v^*)\leq\pi(u)$ for all $u \in \S$. Formally,
\begin{eqnarray*}
 &  & \Pr_{\pi\in\Pi_{\S}}\left[\forall u\in\S,\pi(v^{*})\leq\pi(u)~|~\pi(v^{**})\leq\pi(v^{*})\right]\\
 & = & \Pr_{\pi\in\Pi_{\S}}\left[\forall u\in\S,\pi(v^{*})\leq\pi(u)\mbox{ and }\pi(v^{**})\leq\pi(u)~|~\pi(v^{**})\leq\pi(v^{*})\right]\\
 & = & \frac{\Pr_{\pi\in\Pi_{\S}}\left[\forall u\in\S,\pi(v^{*})\leq\pi(u)\mbox{ and }\pi(v^{**})\leq\pi(u) \mbox{ and }    \pi(v^{**})\leq\pi(v^{*})     \right]}{\Pr_{\pi\in\Pi_{\S}}\left[\pi(v^{**})\leq\pi(v^{*})\right]}\\
 & \leq & \frac{\Pr_{\pi\in\Pi_{\S}}\left[\forall u\in\S,\pi(v^{*})\leq\pi(u)\mbox{ and }\pi(v^{**})\leq\pi(u)\right]}{\Pr_{\pi\in\Pi_{\S}}\left[\pi(v^{**})\leq\pi(v^{*})\right]}
\end{eqnarray*}
To bound the above expression, we separate our discussion into three possible cases. In the first $v^{**} \neq v^*$ and $v^{**} \in P$. The value of the expression is clearly 0 in this case. In the second we have $v^* = v^{**}$ and according to Equation~\eqref{eq:vs_min} we have that the quantity is bounded by $1/|P|$. The last case is the one where $v^{**} \notin P$. Here, because $v^* \in P$ and $v^{**} \notin P$ we have that the events of $\pi(v^*)$ being the minimal in $\{\pi(u)\}_{u\in\S}$ and $\pi(v^{**})$ being the smaller than each of the elements of $\{\pi(u)\}_{u\in\S}$ are independent for uniform $\pi \in \Pi_{\S}$. This is due to the first event being dependent of the inner order inside $\S$ and the second being independent of the same inside order.
Hence,
\begin{eqnarray*} 
 &  & \frac{\Pr_{\pi\in\Pi_{\S}}\left[\forall u\in\S,\pi(v^{*})\leq\pi(u)\mbox{ and }\pi(v^{**})\leq\pi(u)\right]}{\Pr_{\pi\in\Pi_{\S}}\left[\pi(v^{**})\leq\pi(v^{*})\right]} \\
  & = & \frac{\Pr_{\pi\in\Pi_{\S}}\left[\forall u\in\S,\pi(v^{**})\leq\pi(u)\right]}{\Pr_{\pi\in\Pi_{\S}}\left[\pi(v^{**})\leq\pi(v^{*})\right]} \cdot \Pr_{\pi\in\Pi_{\S}}\left[\forall u\in\S,\pi(v^{*})\leq\pi(u)\right] \\
 & \leq & \Pr_{\pi\in\Pi_{\S}}\left[\forall u\in\S,\pi(v^{*})\leq\pi(u)\right]  \leq 1/|P|
\end{eqnarray*}
\vspace{-0.2in}

\end{proof}
\vspace{-0.1in}

Lemma~\ref{lem:SandS'} and Lemma~\ref{lem:one-over-P} immediately lead to Theorem~\ref{thm:ES-const}. Also, as an immediate corollary of Theorem~\ref{thm:ES-const} we get

\begin{corollary}
A direct distributed implementation of Algorithm~\ref{alg:template} has, in expectation, both a single adjustment and round, in both the synchronous and asynchronous models.
\end{corollary}
%

\section{A Constant Broadcast Implementation} \label{sec:dist_impl}

Theorem~\ref{thm:ES-const} promises that the expected number of nodes that need to change their output according to our template algorithm is $1$. However, a direct implementation of the template in Algorithm~\ref{alg:template} in a dynamic distributed setting may require a much larger broadcast complexity because it may be the case that a node needs to change its state several times until the \MIS invariant holds at all nodes. This is because a node can be in more than a single set $S_i$, as discussed in the previous section. In such a case, despite the fact that the expected number of nodes in $S$ is a constant, it may be that the expected number of state changes is much larger. Specifically, in a naive implementation, the number of broadcasts may be as large as $|S|^2$. Hence, although $\E[|S|]=1$, the expected number of broadcasts may be as large as $n$.

We thus take a different approach for implementing the template in Algorithm~\ref{alg:template}, in the \emph{synchronous} setting, where each node waits until it knows the maximal $i$ for which it belongs to $S_i$, and changes it state only once. This allows to obtain, for almost all of the possible topology changes a constant broadcast complexity at the cost of
a constant, rather than single, round complexity.

In order to implement the random permutation $\pi$ we assume each node $v \in V$ has a uniformly random and independent ID $\ID_v \in [0,1]$. We will maintain the property that each node has knowledge of its $\ID$ value and those of its neighbors.
We describe our algorithm in Algorithm~\ref{alg:state_protocol}. This directly applies to the following topology changes: edge-insertion, graceful-edge-deletion, abrupt-edge-deletion, graceful-node-deletion and node-unmuting. An extension of the analysis is provided in Subsection~\ref{subsec:abrupt} for the case of an abrupt node deletion, and a slight modification is provided in Subsection~\ref{subsec:node-insertion} for the case of node-insertion. The following summarizes the guarantees of our implementation, and is proven in Lemmas~\ref{lem:bound_by_S},~\ref{lem:node-insertion}, and~\ref{lem:abrupt}.
.

\begin{thm} \label{thm:main_formal}
There is a complete fully dynamic distributed MIS algorithm which requires in expectation a single adjustment and $O(1)$ rounds for all topology changes. For edge insertions and deletions, graceful node deletion, and node unmuting, the algorithm requires $O(1)$ broadcasts, for an abrupt deletion of a node $v^*$ it requires $O\left(\min\{\log(n), d(v^*)\}\right)$ broadcasts, and for an insertion of a node $v^*$ it requires $O(d(v^*))$ broadcasts, in expectation.
\end{thm}

In the algorithm a node may be in one of four states: $\M$ for an MIS node, $\NM$ for a non-MIS node, $\C$ for a node that may need to change from $\M$ to $\NM$ or vice-versa, and $\R$ for a node that is ready to change. We will sometimes abuse notation and consider a state as the set of nodes which are in that state. Our goal is to maintain the \MIS invariant.

\begin{algorithm}[h!]
\caption{MIS Algorithm for node $v$}
\begin{algorithmic}[1]
\STATE  $v \in \M$: If some $u \in I_\pi(v)$ changes to state $\C$, change state to $\C$.
\STATE  $v \in \NM $: If some $u  \in I_\pi(v)$ changes to state $\C$ and all other $w  \in I_\pi(v)$ are not in $\M$, change state to $\C$.
\STATE  $v \in \C $: If (1) all neighbors $u$ with $\pi(v) < \pi(u)$ are not in state $\C$ and (2) $v$ changed to state $\C$ at least $2$ rounds ago, change state to $\R$.
\STATE $v \in \R $: If all  $u  \in I_\pi(v)$ are in states $\NM$ or $\M$, change state to $\M$ if all $u  \in I_\pi(v)$ are in $\NM$, and change state to $\NM$ otherwise.
\end{algorithmic}
\label{alg:state_protocol}
\end{algorithm}

Any change of state of a node is followed by a broadcast of the new state to all of its neighbors. We now define our implementation as a sequence of state changes.
When a topology change occurs at node $v^*$, if the \MIS invariant still holds then $v^*$ does not change its state and algorithm consists of doing nothing. Otherwise, $v^*$ changes its state to $\C$.


From states $\M$ or $\NM$, a node changes to state $\C$ when it discovers it is in the set $S$ of influenced nodes, as defined in Equation~\eqref{eq:defS}.  From state $\C$, a node $v$ changes to state $\R$ when (1) none of its neighbors $u$ for which $\pi(v) < \pi(u)$ are in state $\C$ and (2) $v$ changed its state to $\C$ at least two rounds ago. Finally, from state $\R$ a node $v$ returns to states $\M$ or $\NM$ when all of its neighbors $u$ for which $\pi(u) < \pi(v)$ are in states $\M$ or $\NM$.
In order to bound the complexity of the algorithm we first show that every node can change from state $\R$ to either $\M$ or $\NM$ at most once.

\begin{lem} \label{lem:Ronce}
In Algorithm~\ref{alg:state_protocol}, a node $u$ changes its state from $\R$ to another state at most once.
\end{lem}
\begin{proof}
First, note that every $u \notin S$ never changes its state.
Consider a node $u$ changing its state from $\R$ to either $\M$ or $\NM$. Since $u$ changes from state $\R$, if $u \neq v^*$  then it must have a neighbor $w \in I_\pi(u)$ that was in state $\C$, changed to state $\R$ and then changed to $\M$ or $\NM$. It follows that $v^*$ must be the first node to change its state from $\R$ to $\M$ or $\NM$. This event occurs only when all neighbors $u$ of $v^*$ are not in $\C$, which in turn can happen only when all neighbors of each such $u$ with higher $\pi$ value have changed from $\C$ to $\R$ at least once. But, since no node could have changed its state from $\R$ to another state before $v^*$ has done so, we have that when $v^*$ changes its state from $\R$ to another, all $u \in S$ are in state $\R$.

In particular, we have that at the round of the first change of a node from $\R$ to another state, there are no nodes in state $\C$. Since a node can only change to state $\C$ due to a neighbor at state $\C$ we have that any node changing its state from $\R$ to $\M$ or $\NM$ will not change its state again, thus proving our claim.
\end{proof}

\begin{lem}\label{lem:bound_by_S}
For edge-insertion, graceful-edge-deletion, abrupt-edge-deletion, graceful-node-deletion and node-unmuting, Algorithm~\ref{alg:state_protocol} requires in expectation a single adjustment, $O(1)$ rounds, and $O(1)$ broadcasts.
\end{lem}
\begin{proof}
Since only nodes in $S$ can change their outputs, the number of adjustments is bounded by $|S|$, and hence is $1$ in expectation, by Theorem~\ref{thm:ES-const}.
According to Lemma~\ref{lem:Ronce}, if a node changes its state then it does so exactly three times. First it changes from either $\M$ or $\NM$ to $\C$, then it changes to $\R$, and finally it changes to either $\M$ or $\NM$ again. Since only nodes in $S$ change their states and since the round and broadcast complexities are clearly bounded by the number of state changes plus 1 (due to the forced waiting round before changing from $\C$ to $\R$), the claim follows.
\end{proof}

\subsection{Node and Edge Insertion}
\label{subsec:node-insertion}
When $v^*$ is inserted, or an edge $(v^*,v^{**})$ is inserted, we assume that all new pairs of neighbors are notified that they are now connected, along with each other's ID. In a setting where this is not the case, we make the following adjustment, before applying Algorithm~\ref{alg:state_protocol}.

When $v^*$ is inserted, in the first round, $v^*$ broadcasts its random $\ID$ value and a temporary state $\NM$ to its neighbors. In the second round, the neighbors of $v^*$ broadcast their states and their $\ID$ values. Now, it may be the case that the \MIS invariant does not hold at $v^*$, but it still holds for any other node in the graph. We now execute Algorithm~\ref{alg:state_protocol}. The expected number of adjustments hence remains $1$, the number of rounds increases by two and is therefore still $O(1)$, and the number of broadcasts is now bounded by the degree $d(v^*)$.

When an edge $(v^*,v^{**})$ is inserted, in the first round, $v^*$ and $v^{**}$ broadcast their random $\ID$ value and state. Now it may be the case that the \MIS invariant no longer holds at $v^*$, but it still holds for any other node in the graph. We then execute Algorithm~\ref{alg:state_protocol}. The expected number of adjustments hence remains $1$, the number of rounds increases by one and is therefore still $O(1)$, and the number of broadcasts is still bounded by $O(1)$.

We therefore get the following:
\begin{lem}\label{lem:node-insertion}
For a node-insertion of a node $v^*$, our algorithm requires in expectation a single adjustment, $O(1)$ rounds, and $O(d(v^*))$ broadcasts.
For an edge-insertion of an edge $(v^*,v^{**})$, our algorithm requires in expectation a single adjustment, $O(1)$ rounds, and $O(1)$ broadcasts.
\end{lem}

\subsection{Abrupt Node Deletion}
\label{subsec:abrupt}
Consider a node $v^*$ that is abruptly deleted. We denote by $v^*_1,v^*_2,\ldots,v^*_x$ the set $S_1=S(\Gold,\Gnew,\pi,v^*)_1$.
We execute Algorithm~\ref{alg:state_protocol}, where in the first round, every $v^*_i$, $1\leq i \leq x$ changes its state to $\C$ (instead of having $v^*$ broadcast its change state to $\C$).
It is straightforward to verify that despite the above modification, only nodes in $S=S(\Gold,\Gnew,\pi,v^*)$ can change to state $\C$ throughout the execution. With this modification, a node may change to state $\C$ more than once. However, we show the amount of times this can happen is bounded by  both the degree of $v^*$ and by $\log(n)$.

%

\begin{lem}\label{lem:abrupt:rounds}
The algorithm completes after at most $3|S|+2$ rounds.
\end{lem}
\begin{proof}
Consider a node $v$ that changes to state $\C$ in round $t$. If $t>1$, then there is a node $u$, for which $\pi(u) < \pi(v)$, that changes to $\C$ in round $t-1$. By induction, it follows that there exists a path $(v_1,v_2,\ldots,v_t)$, where $v_1=v^*_i$ for some $1\leq i \leq x$ and $v_t = v$, and in addition $\pi(v_i) < \pi(v_{j+1})$ for all $j<t$. Since the latter implies that the nodes are distinct, we have that $t \leq |S|$, meaning that after round $|S|$ no node changes to $\C$. This gives that after an additional round, nodes begin to change to $\R$. The same argument now gives that no node changes to $\R$ after round $2|S|+1$. After an additional round nodes start changing from $\R$, and therefore, using the same argument again, we have that no node changed its state from $\R$ after round $3|S|+2$.
To complete the proof, we argue that indeed this procedure progresses, hence eventually all nodes are in either $\M$ or $\NM$, with the \MIS invariant holding.
This holds since two consecutive rounds with no state changes implies that the algorithm terminated: After a round with no state change, if there are nodes in state $\C$ then the node with the maximal $\pi$ order among them changes to $\R$. Otherwise, if there are nodes in state $\R$ then the node with the minimal $\pi$ order among them changes to either $\M$ or $\NM$.
\end{proof}
\begin{lem}
\label{lem:abrupt-t}
If $v$ changes from either $\M$ or $\NM$ to $\C$ at rounds $t$, it does not change to neither $\M$ or $\NM$ again before round $3t+1$. Further, each change of $v$ to $\C$ can be associated with a different node $v^*_i$, for some $1\leq i \leq x$.
\end{lem}
\begin{proof}
We prove the lemma by induction on $t$, where the base case for $t=1$ trivially holds, as a node needs to make 3 changes. For $t>1$, as in the previous lemma, there exists a path $(v_1,v_2,\ldots,v_t)$, where $v_1=v^*_i$ for some $1\leq i \leq x$ and $v_t = v$, and in addition $\pi(v_j) < \pi(v_{j+1})$ for all $j<t$.

Let $t_{\R} \geq t+2$ be the first round in which $v=v_t$ changes to $\R$. Notice that until that round, $v_{t-1}$ does not change to $\R$, and the same holds inductively for $v_j$ for all $1\leq j \leq t-1$.
Let $t_{\R}'$ be the first round in which $v_1$ changes to $\R$, and let $t_S$ be the first round in which $v_1$ changes to either $\M$ or $\NM$. It follows that $t_{\R}' \geq t_{\R}+t-1 \geq 2t+1$, and hence $t_S \geq 2t+2$.


Now, let $t_S'$ be the first round in which $v=v_t$ changes to either $\M$ or $\NM$. In round $t_S'-1$ the node $v_{t-1}$ must be in either $\M$ or $\NM$, and inductively we have that at round $t_S'-t+1$, the node $v_1$ is in either $\M$ or $\NM$. This implies that $t_S' - t+ 1 \geq t_S \geq 2t+2$, giving $t_S' \geq 3t+1$ as required.

Further, when $v$ changes its state from $\R$ to $\M$ or $\NM$, the state of $v_i^*$ is already $\M$ or $\NM$. This implies that if $v$ change its state to $\C$ again in time $t'>t$, then it is due to a change in $v_{i'}^*$, for some $i'\neq i$, $1\leq i' \leq x$. (Although it is possible that the path from $v_{i'}^*$ to $v$ goes through $v_i^*$ again.)
This means that each change of $v$ to $\C$ can be associated with a different node $v_i$, for some $1\leq i \leq x$.
\end{proof}

\begin{lem}\label{lem:abrupt}
For an abrupt deletion of a node $v^*$, our algorithm requires in expectation a single adjustment, $O(1)$ rounds, and $O(\min\{\log(n),d(v^*)\})$ broadcasts.
\end{lem}
\begin{proof}
Now, the adjustment complexity is bounded by $|S|$ and thus by Theorem~\ref{thm:ES-const} it is $1$ in expectation. By Lemma~\ref{lem:abrupt:rounds} the round complexity is $O(|S|)$, and thus by Theorem~\ref{thm:ES-const} it is $O(1)$ in expectation. Finally, as a corollary from Lemma~\ref{lem:abrupt-t} we get that the sequence of rounds in which a node $v$ has changed to $\C$ is $t_1,t_2,\ldots, t_r$ with $t_1 \geq 1$, $t_{i+1} \geq 3t_i+1$ and by Lemma~\ref{lem:abrupt:rounds} we have also $t_r < 3|S|+1$. It follows that $r \leq \log_3(O(|S|))$. Moreover, since by Lemma~\ref{lem:abrupt-t} every change to $C$ by a node $v$ can be accounted to a different node $v^*_i$, 
we have that $r \leq x \leq d(v^*)$. This gives that the total number of broadcasts is at most $O(|S|\min\{\log(|S|), d(v^*)\})$. The claim immediately follows as since $\E[|S|] \leq 1$ (Theorem~\ref{thm:ES-const}) and $|S| \leq n$.
\end{proof}

\section{History Independence}
In this section, we define and discuss the \emph{history independence} property.


\vspace{-0.3cm}
\paragraph{Motivation:}
In many well known problems, in addition to computing a feasible solution, it is sometimes required to compute an optimal solution with respect to a given objective function. For example, one may wish to find an MIS that has a maximal cardinality. Usually, obtaining optimal solutions, or even good approximate solutions, is NP-hard and therefore, we cannot expect to obtain such solutions in the distributed model. Furthermore, typically such optimization algorithms are tailored to a specific objective function, and hence it is required to handle each objective separately. Moreover, in a dynamic setting, it is more cumbersome to analyze the guarantees for the value of the objective function.

Therefore, although we do not wish to consider any specific objective function in the problem description, we do wish that the adversary will lack the ability to choose a feasible solution as she pleases (in this case we may assume she chooses the worst solution). In fact, we require that the adversary will even not be able to \emph{bias} the output of the algorithm towards any specific solution. Formally, we define our requirement as follows.

\begin{defn}
\label{def:history}
Let $A$ be an algorithm for maintaining a combinatorial graph structure $P$ in a dynamic distributed setting. We say that $A$ is \emph{history independent} if given a graph $G$, the output $P$ of $A$ is a random variable whose distribution depends only on $G$, and does not depend on the history of topology changes that constructed $G$.
\end{defn}

\vspace{-0.4cm}
\paragraph{Composability:}
Another advantage of history independent algorithms is that they compose nicely. For example, given a history independent algorithm $A$ for MIS, we can simulate $A$ on the line graph $L(G)$ for obtaining a history independent algorithm $B$ for maximal matching. Alternatively, we can simulate $A$ on a graph $G'=f(G)$ in which every node $v$ in $G$ corresponds to a clique of size $\Delta+1$ in $G'$, and every edge in $G$ corresponds to a matching between the appropriate cliques. This standard reduction by~\cite{Luby86} is known to give an algorithm $C$ for $(\Delta+1)$-coloring in $G$, and since $A$ is history independent then so is $C$. We note that performing the above simulations is non-trivial, as the simple topological changes in $G$ translate into more complex changes in $L(G)$ or $G'$, yet, the challenges are only technical and require no additional insights, and hence we omit the details.

\paragraph{Examples:}
It is straightforward to see that Algorithm~\ref{alg:state_protocol} is a history independent MIS algorithm, because its output for any graph $G$ is identical to the output of the greedy sequential algorithm on $G$, regardless of the topology changes that resulted in $G$. To further exemplify what this important property gives, consider the following examples, in which we compare executions of history independent algorithms to worst case solutions. Although history dependent algorithms do not necessarily yield the worst solution, the \emph{natural} history dependent algorithm indeed yields the worst solution in these examples. Here we think of the natural algorithm as the greedy algorithm that gives every new node or edge the best value that is possible without making any global changes. For this natural algorithm, one can easily verify that for any feasible output there is a pattern of topology changes that can force the algorithm to produce it.


\textbf{Example: MIS in a Star.}
Assume that an adversary controls the topology changes in the graph and chooses them so that a graph $G_{star}$ is created, where $G_{star}$ is a star on $n$ nodes. Since our MIS algorithm simulates random greedy, there is a probability of $1/n$ that the center of the star has the lowest order among all nodes, in which case it is the only node in the MIS. With the remaining probability of $1-1/n$, a different node has the lowest order, which results in the MIS being all nodes except the center, which is the largest MIS possible in $G_{star}$. The expected size of the resulting MIS is therefore linear in $n$, implying that it is within a constant factor of the size of the \emph{maximum} independent set (maximal cardinality independent set). For comparison, recall that the worst-case MIS in a star is the center alone, and its size is $1$.

\textbf{Example 2: Maximal matching of many $3$-paths.} Assume an adversary constructs a graph $G_{3paths}$, which contains $n/4$ disjoint paths of length $3$ edges. Since our maximal matching algorithm simulates a random greedy MIS algorithm on the line graph $L(G_{3paths})$, we have that for every $3$-path independently, with probability $2/3$ its matching is of size $2$ and with probability $1/3$ its matching is of size $1$. Therefore, the expected size of the matching we obtain is $5n/12$. For comparison, notice that the worst-case maximal matching in $G_{3paths}$ has size $n/4$.

\textbf{Example 3: Coloring.}
Regarding coloring algorithms, there is much more room for improvement upon using the standard reduction with our MIS algorithm. As in our MIS algorithm, we would have liked to have an algorithm that simulates the sequential random greedy algorithm for coloring as well, since, for example, it would imply the following.

Assume that an adversary controls the topology changes in the graph and chooses them so that a graph $G=(V,E)$ is created, where $G$ is a bipartite graph on the set of nodes $V= L\cup R$, for $L=\{u_1,\dots,u_{n/2}\}$ and $R=\{v_1,\dots,v_{n/2}\}$. In $E$ we have an edge between the nodes $u_i$ and $v_j$ for every $i \neq j$. Thus, $G$ is a complete bipartite graph minus a perfect matching. If we run a random greedy coloring algorithm, then the first node, say $u_i$ gets the color $1$ when being inserted into the graph. If the next inserted node after $u_i$ is $v_j$, for any $j\neq i$, then it gets the color $2$, and afterwards every node in $L$ gets color $1$ and every node in $R$ gets color $2$. If the next inserted node after $u_i$ is $u_j$, for any $j\neq i$, then it gets the color $1$, and afterwards every node in $L$ gets color $1$ and every node in $R$ gets color $2$. That is, with probability $1-1/n$, we get an optimal $2$-coloring. With the remaining $1/n$ probability we might get a coloring as bad as linear in $\Delta$. This gives that in expectation, we get a coloring whose palette size is a constant factor away from optimal.

We can, of course, simulate the random greedy sequential coloring, but the problem is that we pay a cost of $2^\Delta$ adjustments. The way we can do this is to maintain that every color $i$ is an MIS in the graph $G - \{C_{j}\}_{1\leq j < i}$, where $C_j$ is the set of nodes of color $j$. However, a change to a node may cause each node in $S$ to induce $2$ sources of sets for the next color, which would result in a total set of $2^\Delta$ adjustments. Notice that this is much worse compared to what a naive dynamic distributed coloring algorithm would give. It is a curious question whether we can indeed enjoy both worlds here, or whether any lower bound can be proved.

~\\
The above examples illustrate how a property of the output of a history independent algorithm on a graph $G$ can be analyzed as a simple combinatorial problem. This can lead to better guarantees compared to only being able to assume the worst case.

\section{Discussion}
This paper studies computing an MIS in a distributed dynamic setting. The strength of our analysis lies in obtaining that for an algorithm that simulates the sequential random greedy algorithm, the size of the set of nodes that need to change their output is in expectation $1$. This brings the locality of the fundamental MIS problem to its most powerful setting.
%


We believe our work sets the ground for much more research in this crucial setting. Below we discuss some open problems that arise from our work.

An immediate open question is whether our analysis can be extended to cope with more than a single failure at a time.
Second, there are many additional problems that can be addressed in the dynamic distributed setting, especially in the synchronous case. We believe that our contribution can find applications in solving many additional dynamic distributed tasks.

A major open question is whether our techniques can be adapted to sequential dynamic graph algorithms, which constitutes a major area of research in the sequential setting~\cite{Frederickson85,HenzingerK99,HenzingerK95,EppsteinGIN97,EvenS81,Tarjan75,HolmLT98,Thorup00,KapronKM13,RodittyZ12,RodittyZ11,BhattacharyaHI15,DemetrescuI00,HenzingerKN13,BernsteinS15}. A formal definition and description of typical problems can be found in, e.g., ~\cite{Demetrescu2010}. Notice that our algorithms are \emph{fully dynamic}, which means that they handle both insertions and deletions (of edges and nodes). Although our template for finding an MIS can be easily implemented in a sequential dynamic setting, it would come with a cost of at least $O(\Delta)$ for the update complexity in a direct implementation. This is because we would have to access neighbors of the set of nodes analyzed in Theorem~\ref{thm:ES-const}. Our distributed implementation avoids this by having them simply not respond since they do not need to change their output, and hence they do not contribute to the communication. Nevertheless, we believe that our approach may be useful for designing an MIS algorithm for the dynamic sequential setting, and leave this for future research.

\bibliographystyle{abbrv}
\bibliography{references}



\end{document}